\documentclass[journal, a4paper, twoside]{IEEEtran}
\IEEEoverridecommandlockouts

\usepackage[utf8]{inputenc} 
\usepackage[T1]{fontenc}
\usepackage{url}
\usepackage{ifthen}
\usepackage{cite}
\usepackage[cmex10]{amsmath} 
\usepackage{amssymb}
\usepackage{graphicx}
\usepackage{epstopdf}
\usepackage{epsfig}
\usepackage{multirow}
\usepackage{amsthm}
\usepackage{comment}
\usepackage{caption}
\usepackage{enumitem}
\usepackage{amsmath}
\usepackage{color}
\usepackage{bbm}
\usepackage{algorithm}
\usepackage{algpseudocode}
\usepackage{csquotes}
\usepackage{subfiles}
\usepackage{caption, multirow, makecell}
\usepackage{array}
\usepackage{caption}
\usepackage{nicematrix}
\usepackage{cuted}
\usepackage{subcaption}
\usepackage{mathtools}

\DeclarePairedDelimiter\floor{\lfloor}{\rfloor}

\newtheorem{thm}{Theorem}
\newtheorem{lem}{Lemma}
\newtheorem{prop}{Proposition}

\newtheorem{defn}{Definition}

\newtheorem{corollary}{Corollary}

\def\BibTeX{{\rm B\kern-.05em{\sc i\kern-.025em b}\kern-.08em
    T\kern-.1667em\lower.7ex\hbox{E}\kern-.125emX}}

\begin{document}

\title{A Secretive Coded Caching \\for Shared Cache Systems using PDAs

}

\author{\IEEEauthorblockN{ Elizabath Peter, K. K. Krishnan Namboodiri and B. Sundar Rajan\\}
	\IEEEauthorblockA{Department of Electrical Communication Engineering, IISc
		Bangalore, India \\
		E-mail: \{elizabathp,krishnank,bsrajan\}@iisc.ac.in}
}

\maketitle

\begin{abstract}
	This paper considers the secretive coded caching problem with shared caches in which no user must have access to the files that it did not demand. In a shared cache network, the users are served by a smaller number of helper caches and each user is connected to exactly one helper cache. To ensure the secrecy constraint in shared cache networks, each user is required to have an individual cache of at least unit file size. For this setting, a secretive coded caching scheme was proposed recently in the literature (\enquote{Secretive Coded Caching with Shared Caches}, in \textit{IEEE Communications Letters}, 2021), and it requires a subpacketization level which is in the exponential order of the number of helper caches. By utilizing the PDA constructions, we propose a procedure to obtain new secretive coded caching schemes for shared caches with reduced subpacketization levels. We also show that the existing secretive coded caching scheme for shared caches can be recovered using our procedure. Furthermore, we derive a lower bound on the secretive transmission rate using cut-set arguments and demonstrate the order-optimality of the proposed secretive coded caching scheme. 
\end{abstract}


\section{Introduction}

The advent of smart devices accompanied by the rise of on-demand streaming services and content-based applications has led to a dramatic increase in wireless data traffic over the last two decades. Coded caching has been proposed as a promising technique to reduce the traffic congestion experienced during peak hours by exploiting the memory units distributed across the network. The idea of coded caching was first introduced in the work of Maddah-Ali and Niesen \cite{MaN}, which emphasized the benefits and need of the joint design of storage and delivery policies in content delivery networks. In \cite{MaN}, the setting considered is that of a single server having access to a library of $N$ equal length files and connected to $K$ users through an error-free shared link. Each user is equipped with a dedicated cache of size $M \leq N$ files. The caches are populated with portions of file contents during off-peak times without knowing the future demands of the users, and this is called the \textit{placement phase}. The \textit{delivery phase} happens at peak times, during which the users inform their demands to the server, and the server aims to satisfy the users' demands with minimum transmission load over the shared link. Each user is able to recover its demanded files using the received messages and the cache contents. The sum of each transmitted message's length normalized with respect to the file length is defined as the \textit{rate} of the coded caching scheme. The objective of any coded caching problem is to jointly design the placement and delivery phases such that the rate required to satisfy the users' demands is minimum.

The coded caching scheme in \cite{MaN}, which is referred to as the MN scheme henceforth, was shown to be optimal under the constraint of uncoded placement when $N \geq K$ in \cite{YMA}, \cite{WTP}.  In \cite{YMA}, the MN scheme was modified to obtain another scheme that was optimal for the $N < K$ case as well. The coded caching approach has been extended to a variety of settings that include decentralized caching \cite{MaN2}, multi-access network \cite{HKD}, shared cache network \cite{PUE} and many more. In \cite{RPKP}, for the same setting considered in \cite{MaN}, an additional constraint was incorporated, which ensures that no user can obtain any information about the database files other than its demanded file either from the cache contents or the server transmissions. This setup was referred to as private or secretive coded caching in the literature \cite{RPKP}, \cite{RPKP1}. We resort to the latter terminology in this work. In \cite{RPKP}, an achievable secretive coded caching scheme was proposed for both centralized and decentralized settings. The scheme in \cite{RPKP} also guarantees secure delivery against external eavesdroppers. The secure delivery condition was addressed separately in \cite{STC}. The secretive coded caching problem was then extended to other settings that include shared cache networks \cite{MeR} device-to-device networks \cite{ZeY}, combination networks \cite{ZeY1} and considering collusion among users in dedicated cache network \cite{MaS}. We consider the problem of secretive coded caching with shared caches introduced in \cite{MeR}. The shared cache network introduced in \cite{PUE} consists of a server with $N$ equal length files and is connected to $K$ users with the assistance of $\Lambda \leq K$ helper caches as shown in Fig.~\ref{fig:setting}. Each user has access to only one cache, and each cache can serve an arbitrary number of users. But to ensure the secrecy condition in the shared cache network, each user is required to have a dedicated cache of size at least one unit of file, and this is the setting considered in \cite{MeR}.

The centralized secretive coded caching scheme presented in \cite{RPKP} uses the idea of secret sharing \cite{CDN},\cite{Shamir} and is derived from the MN scheme. Hence, the scheme requires a subpacketization level-the number of smaller parts to which a file is split into- which is in the exponential order of $K-1$. Therefore, the scheme requires splitting of finite length files into an exponential number of packets resulting in a scenario where the overhead bits involved in each transmission outnumber the data bits present in it, and this limits the practical applicability of the scheme. In \cite{YCT}, Yan \textit{et al.} showed that combinatorial structures called Placement Delivery Arrays (PDAs) can be used to design coded caching schemes for dedicated cache networks having low subpacketization levels. For dedicated cache setup, secretive coded caching schemes with reduced subpacketization levels were obtained from PDAs in \cite{MeR1}. The exponentially growing subpacketization level with respect to the number of caches is pervasive in shared cache systems as well and is evident from the secretive coded caching scheme proposed for shared cache networks in \cite{MeR}. In \cite{MeR}, the subpacketization level required is in the exponential order of $\Lambda-1$. Even for moderate network sizes, the subpacketization level required by the scheme in \cite{MeR} turns out to be high. Since shared cache networks better capture more realistic settings and maintaining the confidentiality of the data is essential in several applications such as paid subscription services, it is necessary to look for practically realizable secretive coded caching schemes for shared caches. In \cite{PeR}, PDAs were used to derive non-secretive coded caching schemes for shared cache networks having low subpacketization levels than the optimal scheme in \cite{PUE}. In this work, we identify new secretive coded caching schemes for shared caches using PDAs having lower subpacketization requirements than the scheme in \cite{MeR}. Further, we characterize the performance of our scheme by deriving a cut-set based lower bound for the shared cache setting after incorporating the secrecy condition.

\subsection{Contributions}
In this work, we study the secretive coded caching with shared cache networks. Our contributions are summarized below:
\begin{itemize}
	\item A procedure is proposed to obtain new secretive coded caching schemes for shared caches using PDAs. The advantage of our procedure is that it results in schemes with lesser subpacketization level than the corresponding scheme in \cite{MeR} (Section~\ref{sec:secretive}).
	\item A lower bound on the optimal rate-memory trade-off of secretive coded caching with shared caches is derived using cut-set based arguments and characterizes the multiplicative gap between the achievable rate and the lower bound (Section~\ref{privateappendix}).
	\item We also show that the secretive coded caching scheme in \cite{MeR} can be recovered using a PDA. Hence, our procedure subsumes the scheme in \cite{MeR} as a special case (Section~\ref{sec:secretive}).
\end{itemize}

The rest of the paper is organized as follows. In Section~\ref{sec:prelim}, we briefly discuss some of the topics that are relevant for our scheme description. We describe the problem setup and present the main results in Sections~\ref{sec:prob_setup} and~\ref{sec:mainres}, respectively. In Section~\ref{sec:secretive}, we describe the proposed scheme and in Section~\ref{privateappendix}, the lower bound and the order optimality of the scheme are presented. Section \ref{sec:con} summarizes our results.

\textit{Notations}: For a positive integer $n$, $[n]$ denotes the set $\{1,2,\ldots,n\}$. For any set $\mathcal{S}$, $|\mathcal{S}|$ denotes the cardinality of $\mathcal{S}$. Binomial coefficients are denoted by $\binom{n}{k}$, where $\binom{n}{k} \triangleq \frac{n!}{k!(n-k)!}$ and $\binom{n}{k}$ is zero for $n < k$. Bold uppercase and lowercase letters denote matrices and vectors, respectively. For a vector $\mathbf{a}=(a_1,a_2,\ldots,a_n)$, $\mathbf{a}_{\mathcal{S}}$ denotes the vector consisting the elements in $\mathbf{a}$ at positions specified by the elements in the set $\mathcal{S} \subset [n]$. The columns of an $m \times n$ matrix $\mathbf{A}$ is denoted by $\mathbf{a}_1$, $\mathbf{a}_2,\ldots,\mathbf{a}_n$. An identity matrix of size $n$ is denoted as $\mathbf{I}_n$. The finite field with $q$ elements is denoted by $\mathbb{F}_q$.

 \begin{figure}[t!]
  	\begin{center}
  		\captionsetup{justification=centering}
  		\includegraphics[width=\columnwidth]{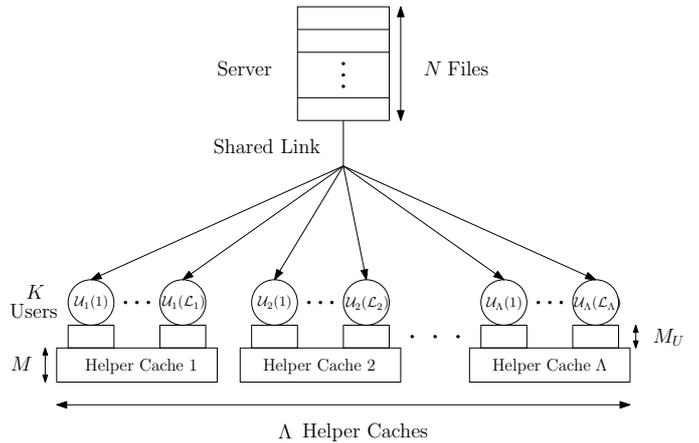}
  		\caption{Problem setting for a shared cache network under secrecy constraint.}
  		\label{fig:setting}
  	\end{center}
  \end{figure}

\section{Preliminaries}
\label{sec:prelim}
In this section, we briefly review PDAs and secret sharing schemes which are required for describing our scheme.

\subsection{Placement Delivery Array (PDA)}
\begin{defn}
	\label{def:pda}
	(\cite{YCT}) For positive integers $K, F, Z$ and $S$, an $F \times K$ array $\mathbf{P}=(p_{j,k})$, $j \in [F]$ and $k \in [K]$, composed of a specific symbol $\star$ and $S$ positive integers $1,2,\ldots, S$, is called a $(K,F,Z,S)$ placement delivery array (PDA) if it satisfies the following three conditions: \\
	\textit{C1}. The symbol $\star$ appears $Z$ times in each column.\\
	\textit{C2}. Each integer occurs at least once in the array.\\
	\textit{C3}. For any two distinct entries $p_{j_1,k_1}$ and $p_{j_2,k_2}$, $p_{j_1,k_1}=p_{j_2,k_2}=s$ is an integer only if
	\begin{enumerate}[label=(\alph*)]
		
		\item $j_1 \neq j_2$, $k_1 \neq k_2$, i.e., they lie in distinct rows and distinct columns, and
		\item $p_{j_1,k_2}=p_{j_2,k_1}=\star$, i.e., the corresponding $2\times2$ sub-array formed by rows $j_1, j_2$ and columns $k_1,k_2$ must be of the following form:\\
		\begin{center}
			
			$\begin{pmatrix}
			s & \star\\
			\star & s
			\end{pmatrix}$
			\hspace{0.3cm}or\hspace{0.3cm}
			$\begin{pmatrix}
			\star & s \\
			s & \star
			\end{pmatrix}$ 
			
		\end{center}

	\end{enumerate}
\end{defn}

Every $(K,F,Z,S)$ PDA corresponds to a coded caching scheme for dedicated cache network with parameters $K,M$ and $N$ as in Lemma~\ref{lemm:1}.

\begin{lem}
	(\cite{YCT}) For a given $(K, F, Z, S)$ PDA $\mathbf{P}=(p_{j,k})_{F \times K}$, a $(K,M,N)$ coded caching scheme can be obtained with subpacketization level $F$ and $\frac{M}{N}=\frac{Z}{F}$ using Algorithm 1. For any distinct demand vector $\mathbf{d}$, the demands of all the users are met with a rate, $R=\frac{S}{F}$.
	\label{lemm:1}
\end{lem}

\begin{algorithm}
	\renewcommand{\thealgorithm}{1}
	\caption{Coded caching scheme based on PDA \cite{YCT}}
	\begin{algorithmic}[1]
		\Procedure{Placement}{$\mathbf{P},W_{[1:N]}$}       
		\State Split each file $W_n$, $n \in [N]$ into $F$ packets: $W_n =\{W_{n,j}: j \in [F]\}$
		\For{\texttt{$k \in [K]$}}
		\State  $\mathcal{Z}_k$ $\leftarrow$ $\{W_{n,j}, \forall n \in [N]$: $p_{j,k}=\star, j \in [F]\}$
		\EndFor
		\EndProcedure
		
		\Procedure{Delivery}{$\mathbf{P},W_{[1:N]},\mathbf{d}$} 
		\For{\texttt{$s \in [S]$}}
		\State Server sends $\underset{\substack{p_{j,k}=s \\j\in [F],  \textrm{\hspace{0.05cm}} k\in[K]}}{\bigoplus}W_{d_k,j}$
		\EndFor    
		\EndProcedure
	\end{algorithmic}
\end{algorithm}

In a $(K,F,Z,S)$ PDA $\mathbf{P}$, the rows represent packets and the columns represent users. For any $k \in [K]$ if $p_{j,k}=\star$, then it implies that the user $k$ has access to the $j^{th}$ packet of all the files. The contents placed in the $k^{th}$ user's cache is denoted by $\mathcal{Z}_k$ in Algorithm 1. If $p_{j,k}=s$ is an integer, then it means that the user $k$ does not have access to the $j^{th}$ packet of any of the files. Condition $C1$ guarantees that all users have access to some $Z$ packets of all the files. According to the delivery procedure in Algorithm $1$, the server sends a linear combination of the requested packets indicated by the integer $s$ in the PDA. Therefore, condition $C2$ implies that the number of messages transmitted by the server is exactly $S$, and the rate achieved is $\frac{S}{F}$. Condition $C3$ ensures the decodability.

\subsection{Secret Sharing Schemes}
 The secretive coded caching schemes proposed so far in the literature rely on non-perfect secret sharing schemes. We also utilize the same in our scheme. The primary idea behind non-perfect sharing scheme is to encode the secret in such a way that accessing a subset of shares does not reveal any information about the secret and only accessing all the shares enable to recover the secret completely. The formal definition of the non-perfect secret sharing scheme is given below. 
 
 \begin{defn}
 	\label{def:secretshare}
  (\cite{CDN}) For a secret $W$ with size $B$ bits and $m < n$, an $(m,n)$ non-perfect secret sharing scheme generates $n$ equal-sized shares $S_1, S_2, \ldots, S_n$ such that accessing any $m$ shares does not reveal any information about the secret $W$ and $W$ can be completely reconstructed from all the $n$ shares. i.e,
   \begin{subequations}
   \begin{align}
     & I(W;\mathcal{S})  =0 \textrm{\hspace{0.25cm}} \forall \mathcal{S} \subseteq \{S_1,S_2,\ldots,S_n\} \textrm{\hspace{0.15cm}} s.t \textrm{\hspace{0.15cm}} |\mathcal{S}| \leq m, \label{share1}\\
     & H(W|S_1,S_2,\ldots,S_n)  = 0. \label{share2}
     \end{align}
   \end{subequations}
 \end{defn}
 
In the non-perfect secret sharing scheme, the size of each share should be at least $\frac{B}{n-m}$ bits \cite{RPKP}. For large enough $B$, there exists non-perfect secret sharing schemes with size of each share being equal to $\frac{B}{n-m}$ bits.

\section{Problem Setup}
\label{sec:prob_setup}
We consider a shared cache network as illustrated in Fig~\ref{fig:setting}. There is a central server with a library of $N$ independent files $W_{[1:N]}=\{W_1,W_2,\ldots,W_N\}$, each of size $B$ bits and is uniformly distributed over $[2^{B}]$. The server is connected to $K$ users through an error-free broadcast link, and there are $\Lambda $ helper caches, each of size equal to $M$ files. Each user gets connected to one helper cache and there is no limit on the number of users served by each helper cache. Further, each user has a dedicated cache of size $M_U$ files. The network operates in four phases as in \cite{MeR}:

\begin{enumerate}[label=(\alph*)]
	\item \textit{Helper Cache Placement phase}: Let the contents stored in the $\lambda^{th}$ helper cache be denoted as ${Z}_{\lambda}$, where $\lambda \in [\Lambda]$. The server fills each of the helper caches with functions of the library files $W_{[1:N]}$ and some randomness $V$ of appropriate size, such that 
	\begin{equation}
	I(W_{[1:N]};Z_{\lambda}) = 0 , \textrm{\hspace{0.05cm}} \forall \lambda \in [\Lambda].
	\label{plsec}
	\end{equation}
	
	Equation \eqref{plsec} implies that  no user is able to retrieve any information regarding any of the files from the cache contents that it gets access to. The placement is carried out without knowing the future demands of the users and their association to the caches and also, it satisfies the memory constraint at each helper cache. Let ${Z}_{[1:\Lambda]} = (Z_1,Z_2,\ldots,Z_{\Lambda})$ denote the contents stored in all the $\Lambda$ helper caches.

	\item \textit{User-to-cache association phase}:
	In this phase, each user gets connected to one of the helper caches and the set of users assigned to cache $\lambda \in [\Lambda]$ is denoted as $\mathcal{U}_{\lambda}$. The overall user-to-cache association is represented as $\mathcal{U}=\{\mathcal{U}_1,\mathcal{U}_2,\ldots,\mathcal{U}_{\Lambda}\}$. All these disjoint sets together form a partition of the set of $K$ users and this association of users to helper caches is independent of the cached contents and the subsequent demands. 
	For any user-to-cache association $\mathcal{U}$, the association profile $\mathcal{L}$ describes the number of users accessing each cache. Therefore, $\mathcal{L}=\{\mathcal{L}_1,\mathcal{L}_2,\ldots,\mathcal{L}_{\Lambda}\}$ where $\mathcal{L}_i =|\mathcal{U}_i|$ and  $\sum_{i=1}^{\Lambda}\mathcal{L}_{i}=K$.  Without loss of generality, assume that $\mathcal{L}_i \geq \mathcal{L}_j$ $\forall i\leq j$ and each $\mathcal{U}_i$ to be an ordered set. Each user in $\mathcal{U}_i$ is indexed as $\mathcal{U}_i(j)$. Several user-to-cache associations result in the same $\mathcal{L}$. Therefore, each $\mathcal{L}$ represents a class of $\mathcal{U}$. Let the helper cache accessed by any user $k \in [K]$ be denoted as $\lambda_k$. For any two users $k_1$ and $k_2$, if $\lambda_{k_1} = \lambda_{k_2}$ then users $k_1$ and $k_2$ are accessing the same cache.
		
	\item \textit{User Cache Placement phase}: Once $\mathcal{U}$ is known to the server, there is an additional phase where the server fills each of the user's dedicated cache with random keys satisfying the memory constraint. The contents stored in the $k^{th}$ user's cache is denoted as $\mathcal{Z}_k$ and $\mathcal{Z}_{[1:K]}=(\mathcal{Z}_1, \mathcal{Z}_2,\ldots, \mathcal{Z}_{K})$ denotes the set of all users' dedicated cache contents. User $k \in [K]$ having access to $\mathcal{Z}_k$ and ${Z}_{\lambda_k}$ should not get any information about $W_{[1:N]}$. That is, 
	\begin{equation}
	 I(W_{[1:N]};{Z}_{\lambda_k},\mathcal{Z}_k) = 0 , \textrm{\hspace{0.05cm}} \forall k \in [K].
	 \label{ussec}
	\end{equation}
     
     \item \textit{Delivery Phase}: In this phase, each user demands one of the $N$ files. The indices of the demanded files are denoted by random variables. Let $\mathcal{D}_k$ be a random variable denoting the $k^{th}$ user's demand. Then, $\mathcal{D}=\{\mathcal{D}_1,\mathcal{D}_2,\ldots,\mathcal{D}_k\}$ is a set of $K$ independent random variables, each uniformly distributed over the set $[N]$. Let $\mathbf{d}=(d_1,d_2,\ldots,d_K)$ be a realization of $\mathcal{D}$. Upon receiving the demand vector $\mathbf{d}$, the server makes a transmission $X$ of size $R_sB$ bits over the shared link to the users, where $X$ is a function of the association profile $\mathcal{L}$, $W_{[1:N]}$, $Z_{[1:\Lambda]}$ and $\mathcal{Z}_{[1:K]}$. Each user $k \in [K]$ must be able to decode its demanded file $W_{d_k}$ using the transmission $X$ and its available cache contents $Z_{\lambda_k}$ and $\mathcal{Z}_k$  and should not obtain any information about the remaining $N-1$ files. That is,
 
       \begin{align}
         H(W_{d_k}|X,Z_{\lambda_k},\mathcal{Z}_k) & = 0  \textrm{\hspace{0.2cm}}\forall k \in [K], \textrm{\hspace{0.2cm}and}  \label{decod}\\
         I(W_{[1:N] \backslash d_k}; X,Z_{\lambda_k},\mathcal{Z}_{k}) & =0 \textrm{\hspace{0.2cm}}\forall k \in [K]. \label{secrecy}
         \end{align}
     
  \end{enumerate}
       
       For a given association profile $\mathcal{L}$, the worst-case rate corresponds to 
       $\underset{\mathcal{D}\in[N]^{K}}{\textrm{max}}R_s$. We aim to minimize the worst-case rate $R_s$ under the decodability and secrecy conditions mentioned in \eqref{decod} and \eqref{secrecy}, respectively.
       
       \begin{defn}
       	 For the above shared cache setting, a memory-rate pair $((M,M_U),R_s)$ is said to be secretively achievable if there exists a scheme for the memory point $(M,M_U)$ that satisfies the decodability condition in \eqref{decod} and the secrecy condition in \eqref{secrecy} with a rate less than or equal to $R_s$ for every possible realization of $\mathcal{D}$. The optimal rate-memory trade-off under secrecy condition is defined as 
       	 \begin{equation*}
       	 \begin{aligned}
       	  R_s^{*}&(M,M_U) = \\&inf\{R_s: ((M,M_U),R_s) \textrm{\textit{\hspace{0.1cm}is secretively achievable}}\}.
       	  \end{aligned}
       	 \end{equation*}
       \end{defn}

\section{Main Results}
\label{sec:mainres}
Before presenting the main results, we first discuss the relevance of the dedicated user cache in our setting and show that the user cache must have a size of at least one file to ensure secrecy in any achievable coded caching scheme for shared cache system.

In a shared cache network, several users will be sharing the same cache contents, hence multicasting opportunities cannot be created amongst the users accessing the same helper cache. Consider user $k$ connected to the helper cache $\lambda$. The transmissions that are useful for the $k^{th}$ user can be  decoded by the remaining users in $\mathcal{U}_{\lambda}$.
Therefore, to ensure secrecy for the file content that user $k$ has requested against the users in $\mathcal{U}_{\lambda}$, the transmissions need to be encrypted using one-time pads that are known to only user $k$ and unknown to other users in $\mathcal{U}_{\lambda}$. To store these random keys, each user needs a dedicated memory unit in addition to the helper cache that it is accessing. As mentioned earlier, each user cache has a capacity to store $M_U$ files, and the condition $M_U \geq 1$ needs to be satisfied to achieve a secretive coded caching scheme for shared caches. The formal proof of it is given below.

Consider a cache $\lambda \in [\Lambda]$ which has more than one user connected to, say $\lambda=1$. Let $\mathbf{d}=(d_1,d_2,\ldots,d_{\mathcal{L}_1},0,\ldots,0)$ be a demand vector where only the users in $\mathcal{U}_1$ demand a file, and let
$X$ be the corresponding transmission made by the server. Choose a user $k \in \mathcal{U}_1$. Then,

\begin{subequations}
  \begin{align}
    \mathcal{L}_1B & \leq H(W_{\mathbf{d}}) \notag \\
                   & \leq I(W_{\mathbf{d}}; X, Z_1, \mathcal{Z}_{[1:\mathcal{L}_1]} ) + H(W_{\mathbf{d}}| X, Z_1, \mathcal{Z}_{[1:\mathcal{L}_1]}) \notag \\
                   & \leq I(W_{\mathbf{d}}; X, Z_1, \mathcal{Z}_{[1:\mathcal{L}_1]} ) \label{u1}\\
                   & \leq I(W_{\mathbf{d}}; X, Z_1, \mathcal{Z}_{k}) \textrm{\hspace{0.1cm}} + \notag \\
                   & \textrm{\hspace{1.5cm}} I(W_{\mathbf{d}}; \mathcal{Z}_{[1:\mathcal{L}_1]\backslash \{k\}} | X, Z_1, \mathcal{Z}_{k} ) \label{u2}\\
                   & \leq I(W_{\mathbf{d}\backslash \{d_k\}}; X, Z_1, \mathcal{Z}_{k}) \textrm{\hspace{0.1cm}} + \notag \\
                   & \quad \quad \quad I(W_{{d_k}}; X, Z_1, \mathcal{Z}_{k} | W_{\mathbf{d}\backslash \{d_k\}}) \textrm{\hspace{0.1cm}} +  \notag \\
                   &  \textrm{\hspace{1.5cm}}  I(W_{\mathbf{d}}; \mathcal{Z}_{[1:\mathcal{L}_1]\backslash \{k\}} | X, Z_1, \mathcal{Z}_{k} ) \label{u3} \\
                   & \leq I(W_{{d_k}}; X, Z_1, \mathcal{Z}_{k} | W_{\mathbf{d}\backslash \{d_k\}}) \textrm{\hspace{0.1cm}} + \notag \\
                   &  \textrm{\hspace{1.5cm}} I(W_{\mathbf{d}}; \mathcal{Z}_{[1:\mathcal{L}_1]\backslash \{k\}} | X, Z_1, \mathcal{Z}_{k} ) \label{u4}\\
                   & \leq H(W_{{d_k}}) + H(\mathcal{Z}_{[1:\mathcal{L}_1]\backslash \{k\}}) \notag \\
                   & \leq B + B(\mathcal{L}_1-1)M_U , \notag
    \end{align}
\end{subequations}

\noindent where \eqref{u1} follows from \eqref{decod}, \eqref{u2} and \eqref{u3} follow from the chain rule of mutual information and \eqref{u4} follows from \eqref{secrecy}. Thus, we obtain $M_U \geq 1$. It is sufficient to consider $M_U$ as unity as users' individual caches are used only for storing the random keys that are used to encrypt those transmissions in which the user is involved. Therefore, in our further discussion, we fix $M_U = 1$, as taken in \cite{MeR} and the shared caching problem described in Section \ref{sec:prob_setup} is referred to as $(\Lambda, K, M, N)$ shared caching problem henceforth.

The following theorem presents a secretive coded caching scheme for shared caches obtained using PDAs.

\begin{thm}
	\label{thm1}
	For a given $(\Lambda, K, M, N)$ shared caching problem with an association profile $\mathcal{L}=(\mathcal{L}_1, \mathcal{L}_2,\ldots, \mathcal{L}_{\Lambda})$, a secretive coded caching scheme with sub-packetization level $F$ can be derived from a $(\Lambda, F, Z,S)$ PDA $\mathbf{P} =(p_{j,\lambda})_{F \times \Lambda}$ satisfying $\frac{Z}{F}=\frac{M}{M+N}$. The secretively achievable worst-case rate is obtained as
	\begin{equation}
	   R_s = \frac{\displaystyle\sum_{s \in [S]} \mathcal{L}_{\tau_s}}{F-Z}
	   \label{eq:thm1}
	\end{equation}
	\noindent where, $\tau_s \triangleq min\{\lambda \in [\Lambda]: s \in \mathbf{p}_{\lambda}\}$ , $\forall s \in [S]$.
\end{thm}

	The scheme that achieves performance in \eqref{eq:thm1} is presented in Section~\ref{sec:secretive}.

Note that the rate $R_s$ varies according to the association profile $\mathcal{L}$ for a given shared caching problem. For a uniform profile, the rate $R_s$ is minimum and as the profile becomes more and more skewed, the rate increases. This follows from the description of the scheme.

\begin{corollary}
     \label{corr}
	 For a uniform association profile, $\mathcal{L}=(\frac{K}{\Lambda},\frac{K}{\Lambda},\ldots, \frac{K}{\Lambda})$, the worst-case rate becomes
	 \begin{equation}
	   R_s = \frac{K. S}{\Lambda(F-Z)}.
	   \label{eq:corr}
	 \end{equation}
\end{corollary}

\begin{proof}
	When the profile is uniform, $\forall s \in [S]$, $\mathcal{L}_{\tau_s}=\frac{K}{\Lambda}$. Thus, \eqref{eq:thm1} reduces to 
	the expression in  \eqref{eq:corr}  .
\end{proof}

The following theorem provides an information-theoretic lower bound on the rate $R_s$ achievable by any secretive coded caching 
scheme for shared cache networks.

\begin{thm}
	\label{thmsecrecy}
	For any $\mathcal{L}$, $M \geq 0$ and $M_U \geq 1$, the achievable secretive rate for a shared cache system is lower bounded by 
	
	\begin{equation}
	\begin{aligned}
	  & R^{*}_s(M,M_U)  \geq  \\
	 &   \max_{s\in \{1,2,\hdots,\min(\frac{N}{2},K)\}} \small{\frac{s\floor{\frac{N}{s}}-1 -(\lambda_s-1)M-(s-1)M_U}{\floor{\frac{N}{s}}-1}}. \\
	   \label{eq:cutset}
	  \end{aligned}
	\end{equation}
	
\end{thm}

	The proof is given in Section~\ref{privateappendix}.

The lower bound expression in \eqref{eq:cutset} has a parameter $s$, which indicates the number of users under consideration. The term $\lambda_s$ corresponds to the cache to which the $s^{th}$ user is connected to. In our setting, $M_U$ is fixed as unity and we define $R_s^{*}(M)$ as $R_s^{*}(M,M_U)|_{M_U=1}$.

\begin{corollary}
		For any $\mathcal{L}$, $M \geq 0$ and $M_U = 1$, the achievable secretive rate is lower bounded by 
	
	\begin{equation}
	R^{*}_s (M) \geq \max_{s\in \{1,2,\hdots,\min(\frac{N}{2},K)\}} s-\frac{(\lambda_s-1)M}{\floor{\frac{N}{s}}-1}.
	\label{eq:cutset2}
	\end{equation}
	
\end{corollary}

\begin{proof}
	The lower bound in \eqref{eq:cutset2} follows directly from \eqref{eq:cutset} after letting $M_U =1$.
\end{proof}

 When $M=0$, the server generates a set of $K$ independent random keys $\{\mathcal{K}_1, \mathcal{K}_2, \ldots, \mathcal{K}_K\}$, each uniformly distributed over $[2^B]$. Then in the user cache placement phase, the key $\mathcal{K}_k$ is placed in the $k^{th}$ user's cache. That is, $\mathcal{Z}_k=\mathcal{K}_k$. In the delivery phase, the server transmits $W_{d_k} \oplus \mathcal{K}_k$, $\forall k \in [K]$ to satisfy the demands of the users. Thus, we obtain $R_s(M=0) = K$. It is straightforward to see that the conditions in \eqref{decod} and \eqref{secrecy} are satisfied by the above transmissions. Therefore, $R_s(M) \leq K$ for $M \geq 0$. The following theorem demonstrates the order-optimality of the obtained scheme.
 
 \begin{thm}
 	\label{thm3}
 	For $M \geq 0$ and $M_U = 1$, if $N \geq 2K$, the rate achieved by the secretive coded caching scheme obtained from PDAs is within the optimal rate by a factor $\Lambda$ which is a system parameter. i.e.,
 	\begin{equation}
 	 	  1 \leq \frac{R_s(M)}{R^{*}_s(M)} \leq \Lambda.
 	\end{equation}
 \end{thm}

 	 The proof is given in Section~\ref{privateappendix}.

 \section{Secretive Coded Caching Scheme for Shared caches using PDAs}
 \label{sec:secretive}
 In this section, we present a procedure to obtain secretive coded caching scheme for shared caches using PDAs. 
 
 Consider a shared cache network shown in Fig.~\ref{fig:setting}. For the given $(\Lambda, K, M, N)$ shared caching problem, choose a $(\Lambda, F, Z, S)$ PDA $\mathbf{P}$ such that $\frac{Z}{F}=\frac{M}{M+N}$. The four phases involved are described below.
 
 \subsection{Helper Cache Placement Phase} 
 \label{subsec:help_cacheplac}
 The server first splits each file in $W_{[1:N]}$ into $F-Z$ non-overlapping subfiles such that each subfile is of size, $F_s=\frac{B}{F-Z}$ bits. Then, each file is encoded using a $(Z,F)$ non-perfect secret sharing scheme. The $F$ shares of the file $W_n$ are denoted by $S_n = \{S_{n,1}, S_{n,2},\ldots, S_{n,F}\}$, where $n \in [N]$ and $|S_{n,j}|=F_s$, $\forall j \in [F]$. Let $S_{[1:N]}$ denote the set of shares of all the $N$ files.
 
 In the PDA $\mathbf{P}$, the rows represent the shares and the columns represent the helper caches. The placement of the shares in the helper caches is defined by the symbol `$\star$' in the corresponding column. That is,
 \begin{equation}
 Z_{\lambda} = \{S_{n,j}, \forall n \in [N] : {p}_{j,\lambda} = \star, j \in [F] \}
 \end{equation}
 
  By Condition $C1$ in Definition \ref{def:pda}, each helper cache stores some $Z$ shares of all the files such that the memory constraint is satisfied.

 \subsection{User-to-cache Association Phase}
  In this phase, each one of the $K$ users gets connected to one of the helper caches.  Once the user-to-cache association $\mathcal{U}$ and the profile $\mathcal{L}$ are known, construct an array, $\mathbf{G}$ of size $F \times K$  from $\mathbf{P}$ as described in lines $1$-$18$ in Algorithm $2$. The array $\mathbf{G}$ is a generalized placement delivery array defined in \cite{PeR}. In $\mathbf{G}$, the numerical entries are an ordered pair, which come from a subset of $S\times \mathcal{L}_1$, where $S \times \mathcal{L}_1 \triangleq \{ (s,i): s \in [S], i \in [\mathcal{L}_1] \}$. Each column in $\mathbf{G}$ corresponds to a user and the symbol `$\star$' represents the shares that are available to each user. Thus, each user gets access to some $Z$ shares of all the files but no information is gained about any of the $N$ files from these $Z$ shares. This follows from the $(Z,F)$ secret sharing scheme that we employ. 
 
 \subsection{User Cache Placement Phase} 
 \label{subsec:user_cacheplac}
 To ensure the secrecy constraint while retrieving the demanded file, some keys need to be stored privately in each user's cache which are essential for encrypting the transmissions. Since each user wants to decode the remaining $F-Z$ shares of its demanded file, it needs to store $F-Z$ independently and uniformly generated random keys of size $F_s$. Thus, $M_UB = {F_s . (F-Z)} = B $ bits, which is in accordance with the memory constraint assumed for the user cache. The key which is used to encrypt a particular transmission will be stored in all those users' caches which are involved in that transmission. Hence, the number of distinct keys that the server generates is same as the number of distinct ordered pairs present in $\mathbf{G}$ and each key is indexed by an ordered pair (line $20$ of Algorithm $2$). As mentioned above, each user $k \in [K]$ stores $F-Z$ keys which are indexed by the ordered pair $(s,i)$ present in the column $\mathbf{g}_k$ (described in lines $21$-$26$ of Algorithm $2$).

  \begin{algorithm}
 	\renewcommand{\thealgorithm}{2}
 	
 	\caption{User Cache Placement and Construction of $\mathbf{G}$ from the $(\Lambda,F,Z,S)$ PDA $\mathbf{P}$}
 	\begin{algorithmic}[1]
 		\Procedure{Construction of $\mathbf{G}$}{$\mathbf{P},\mathcal{U}$}   
 		\State Obtain $\mathcal{L} \gets (\mathcal{L}_1, \mathcal{L}_2,\ldots, \mathcal{L}_{\Lambda})$, $\mathcal{L}_i \geq \mathcal{L}_j$ $\forall i \leq j$ 
 		.
 		\State Construct $\mathbf{G} = (g_{j,k}),$ $j \in [F]$, $k \in [K]$
 		\State $k \gets 1$ 
 		\For {$\lambda \in [\Lambda]$}
 		\If {$\mathcal{L}_{\lambda}>0$}
 		\For {$i \in [0,\mathcal{L}_{\lambda})$}
 		\State $\mathbf{g}_k = \mathbf{p}_{\lambda}$
 		\For {$j \in [F]$}
 		\If {${g}_{j,k} \neq \star$ } 
 		\State ${g}_{j,k} = ({g}_{j,k},1)$.
 		\State ${g}_{j,k} = {g}_{j,k}+(0,i)$.
 		\EndIf
 		\EndFor
 		\State $k \gets k+1$		   		
 		\EndFor
 		\EndIf
 		\EndFor
 		\EndProcedure

 		\Procedure{User Cache Placement}{$\mathbf{G}, \mathbf{P}, F_s$}  
 		\State For every distinct $(s,i)$ in $\mathbf{G}$, server generates an independent random key $\mathcal{K}_{(s,i)}$, uniformly distributed over $2^{[F_s]}$.
 		\For{$\lambda \in [\Lambda]$}
 		\For{$i \in [\mathcal{L}_{\lambda}]$}
 		\State $k \gets \mathcal{U}_{\lambda}(i)$
 		\State $\mathcal{Z}_k = \underset{s \in [S]}{\bigcup}(\mathcal{K}_{(s,i)}: s \in \mathbf{p}_{\lambda})$
 		\EndFor
 		\EndFor
 		\EndProcedure

 	\end{algorithmic}
 \end{algorithm}

\begin{algorithm}
	\renewcommand{\thealgorithm}{3}
	\caption{Delivery Procedure}
	\hspace*{\algorithmicindent} \textbf{Input}: $\mathbf{G}$, $S_{[1:N]}$, $\mathbf{d}$ 
	\begin{algorithmic}[1]
		\For{\texttt{$s \in [S]$}}
		\For{$i \in [\mathcal{L}_1]$}
		\If {$(s,i)$ exists}
		\State Server sends $\underset{\substack{g_{j,k}=(s,i)\\ j \in [F],\textrm{\hspace{0.1cm}}k\in[K]}}{\bigoplus}S_{d_k,j} \bigoplus \mathcal{K}_{(s,i)}$
		\EndIf
		\EndFor
		\EndFor
	\end{algorithmic}
\end{algorithm}

 \subsection{Delivery Phase}
  In the delivery phase, users' inform their demands to the server. Let $\mathbf{d}=(d_1,d_2,\ldots,d_K)$ be a demand vector. Consider the worst-case scenario where all the demands are distinct. The server transmits a message corresponding to every distinct ordered pair $(s,i)$ in $\mathbf{G}$. Let $m$ denote the number of times $(s,i)$ occurs in $\mathbf{G}$. Assume $g_{j_1,k_1}=g_{j_2,k_2}=\ldots=g_{j_m,k_m}=(s,i)$. Then, the sub-array formed by the rows $j_1,j_2,\ldots,j_m$ and the columns $k_1,k_2,\ldots,k_m$ is equivalent to a scaled identity matrix $\mathbf{I}_m$ up to row or column permutations \cite{PeR} as shown in \eqref{clique}.
  \begin{center}
 	 \begin{equation}
 	\begin{aligned}
 	\begin{pmatrix}
 	(s,i) & \star & \ldots  &  \star \\
 	\star &  (s,i) & \ldots & \star   \\
 	\vdots & \vdots & \ddots & \vdots \\
 	\star & \star & \ldots & (s,i)
 	\end{pmatrix}_{m \times m}
 	\label{clique}
 	\end{aligned}
 	\end{equation} 
 \end{center}
 
 Each user $k_u$, $ u \in [m]$  has the key $\mathcal{K}_{(s,i)}$ and the set of shares $\underset{\substack{i \in [m], i \neq u}}{\bigcup}S_{d_{k_i},j_i}$ wanted by other $m-1$ users. Hence, the server transmits a message of the form $\mathcal{K}_{(s,i)} \oplus \underset{1 \leq u \leq m}{\oplus}S_{d_{k_u},j_u}$  for every distinct $(s,i)$ in $\mathbf{G}$ The delivery procedure is described in Algorithm 3.
 
 \subsection{Decoding}
   The decoding of the shares from the transmissions follows directly from \eqref{clique}. Each user has access to $Z$ shares and  obtains the remaining $F-Z$ shares of its desired file from the transmissions, by using the helper cache contents and the keys that are stored privately in its cache. Hence, each user can retrieve its demanded file from its $F$ shares as mentioned in Definition \ref{def:secretshare}.

\subsection{Proof of Secrecy}
Consider a user $k \in [K]$, and its accessible cache contents $Z_{\lambda_k}$ and $\mathcal{Z}_k$. According to the placement procedure described in Section \ref{subsec:help_cacheplac} and Section \ref{subsec:user_cacheplac}, the helper cache contents $Z_{\lambda_k}$ consist of some $Z$ shares of all the $N$ files and $\mathcal{Z}_k$ is constituted by independently and uniformly generated random keys which are used for one-time padding. By virtue of the $(Z,F)$ secret sharing scheme that is used, the $Z$ shares of a file do not reveal any information about it and the shares of one file are independent of the other. Therefore, having access to all the shares of one file do not convey any information about other files as well. Thus, we obtain:
\begin{align}
  & I(W_{[1:N]};Z_{\lambda_k}, \mathcal{Z}_k) = 0 , \\
 & I(W_{[1:N] \backslash d_k}; (Z_{\lambda_k} \cup S_{d_k}), \mathcal{Z}_k) = 0.
\end{align}

\subsection{Calculation of Rate}
 Now, we calculate the required transmission rate in the worst-case scenario. According to the delivery procedure summarized in Algorithm $3$, there is a transmission corresponding to every distinct $(s,i)$ in $\mathbf{G}$ and each transmission is of size $F_s$ bits. For each $s \in {S}$, the value that $i$ takes is different as it depends on the association profile $\mathcal{L}$. Assume $s \in [S]$ appears $m$ times in the PDA $\mathbf{P}$ and let $\{\lambda_1,\lambda_2,\ldots,\lambda_m\}$ be the set of column indices in which $s$ occurs. Then, the number of transmissions in which $s$ is involved depends on the number of users connected to the most populated cache amongst the above set of $m$ helper caches. The most populated cache in the above set corresponds to the minimum of $\{\lambda_1,\lambda_2,\ldots,\lambda_m\}$ as it is assumed that the helper caches are labelled in a non-increasing order of the number of users accessing it. Thus, we obtain the normalized rate as
 \begin{equation*}
    R_s = \frac{\displaystyle\sum_{s \in [S]} \mathcal{L}_{\tau_s}}{F-Z}
 \end{equation*}
  \noindent where, $\tau_s$ is defined as the minimum column index in which $s$ appears in the PDA $\mathbf{P}$.

This concludes the proof of Theorem \ref{thm1}.\hfill $\blacksquare$

\subsection{Example: $(\Lambda=6,K=21,M=21,N=21)$ shared caching problem with $\mathcal{L}=(6,5,4,3,2,1)$.}
Consider a setting with a server having access to $N=21$ files $W_{[1:21]}=\{W_1,W_2,\ldots,W_{21}\}$, each of size $B$ bits. The server is connected to $K=21$ users, each possessing a cache of size $M_U=1$ file. There are $\Lambda =6$ helper caches, each of size equal to $M=6$ files. For this setting, we start with a $(\Lambda=6, F =4, Z=2, S=4)$ PDA $\mathbf{P}$ given in \eqref{ex:pda}, which satisfies the condition $\frac{Z}{F}=\frac{M}{M+N}$.

\begin{equation}
 \mathbf{P}= \begin{pmatrix}
    \star & \star & \star & 1 & 2 & 3 \\
    \star & 1 & 2 & \star & \star & 4 \\
    1 & \star & 3 & \star & 4 & \star \\
    2 & 3 & \star & 4 & \star & \star
  \end{pmatrix}
  \label{ex:pda}
\end{equation}

 Each file $ W_n \in W_{[1:21]}$ gets splits into $F-Z =2 $ subfiles $\{W_{n,1},W_{n,2}\}$, each of size $F_s = B/2$ bits. Then, each file is encoded using a $(2,4)$ non-perfect secret sharing scheme. To generate the shares of a file, first form an $F \times 1$ column vector comprised of $F-Z$ subfiles and $Z$ independent random keys $\{V_{n,1},V_{n,2},\ldots, V_{n,Z}\}$, each uniformly distributed over $[2^{F_s}]$. Then, pre-multiply it with the parity check matrix of a $(2F,F)$ MDS code over $\mathbb{F}_{2^l}$, where $l$ is sufficiently large such that the $(2F,F)$ MDS code exists. In this example, we consider a $4 \times 4$ Cauchy matrix over $\mathbb{F}_{2^3}$. Thus, the four shares of the file $W_n$, $n \in [21]$ are obtained as follows: 
\begin{align*}
 \begin{bmatrix}
S_{n,1} \\
S_{n,2} \\
S_{n,3} \\
S_{n,4}
\end{bmatrix}=
\begin{bmatrix}
 1 & 6 & 2 & 4 \\
 6 & 1 & 4 & 2 \\
 2 & 4 & 1 & 6 \\
 4 & 2 & 1 & 6
\end{bmatrix}
\begin{bmatrix}
W_{n,1}\\
W_{n,2}\\
V_{n,1}\\
V_{n,2}
\end{bmatrix}
\end{align*}

The contents stored in each helper cache are:

\begin{equation*}
\begin{aligned}
  Z_1 & = \{S_{n,1}, S_{n,2} : n \in [21]\}\\
  Z_2 & = \{S_{n,1}, S_{n,3} : n \in [21]\}\\
  Z_3 & = \{S_{n,1}, S_{n,4} : n \in [21]\}\\
  Z_4 & = \{S_{n,2}, S_{n,3} : n \in [21]\}\\
  Z_5 & = \{S_{n,2}, S_{n,4} : n \in [21]\}\\
  Z_6 & = \{S_{n,3}, S_{n,4} : n \in [21]\}\\
\end{aligned}
\end{equation*}

Let the user-to cache association be $\mathcal{U}=\{\{1,2,3,4,5,6\},$ $\{7,8,9,10,11\},\{12,13,14,15\}\{16,17,18\},\{19,20\},\{21\}\}$ with $\mathcal{L}=(6,5,4,3,2,1)$. The generalized PDA $\mathbf{G}$ is obtained as given in \eqref{ex:gpda}. Each user has access to $2$ shares of each file, but the helper cache contents do not reveal any information about $W_{[1:21]}$ due to the $(2,4)$ non-perfect secret sharing encoding.

\setcounter{MaxMatrixCols}{50}
\begin{figure}
	\begin{equation}
	\mathbf{G}^T = \begin{pmatrix}
	\star &\star &(1,1) &(2,1)\\
	\star &\star &(1,2) &(2,2)\\
	\star &\star &(1,3) &(2,3)\\
	\star &\star &(1,4) &(2,4)\\
	\star &\star &(1,5) &(2,5)\\
	\star &\star &(1,6) &(2,6)\\
	\star &(1,1) &\star &(3,1)\\ 
	\star &(1,2) &\star &(3,2)\\ 
	\star &(1,3) &\star &(3,3)\\ 
	\star &(1,4) &\star &(3,4)\\ 
	\star &(1,5) &\star &(3,5)\\ 
	\star &(2,1) &(3,1) &\star\\
	\star &(2,2) &(3,2) &\star\\
	\star &(2,3) &(3,3) &\star\\
	\star &(2,4) &(3,4) &\star\\
	(1,1) &\star &\star &(4,1)\\
	(1,2) &\star &\star &(4,2)\\
	(1,3) &\star &\star &(4,3)\\
	(2,1) &\star &(4,1) &\star\\
	(2,2) &\star &(4,2) &\star\\
	(3,1) &(4,1) &\star &\star\\
	\end{pmatrix}_{21 \times 4}
	\label{ex:gpda}
	\end{equation}
\end{figure}

\begin{figure*}[t!]
	\begin{subfigure}[t]{0.5\textwidth}
		\centering
		\includegraphics[width=\textwidth]{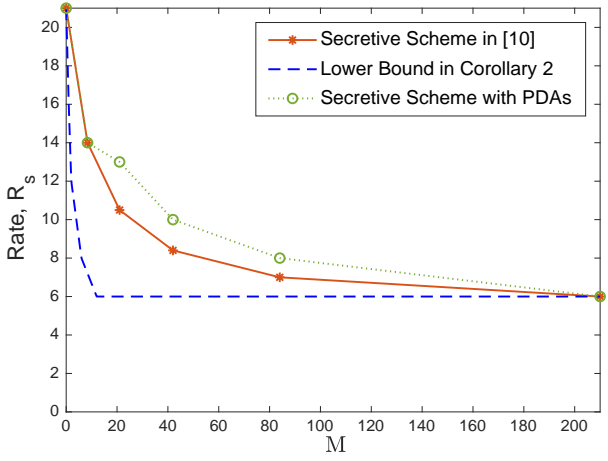}
		\caption{Association profile $\mathcal{L}=(6,5,4,3,2,1)$}
		\label{fig:rateprof6543}
	\end{subfigure}
	\hfill
	\begin{subfigure}[t]{0.5\textwidth}
		\centering
		\includegraphics[width=\textwidth]{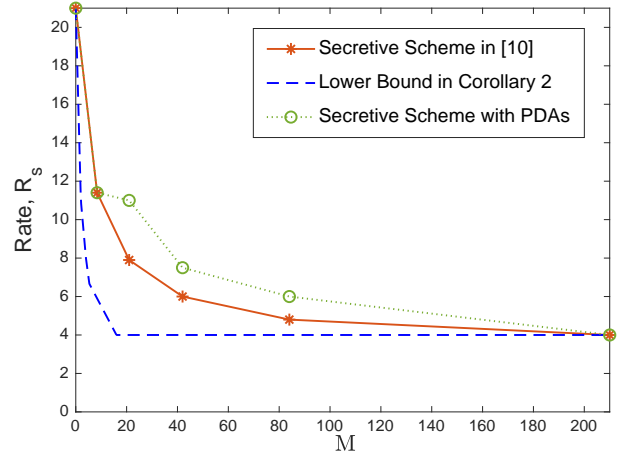}
		\caption{Association profile $\mathcal{L}=(4,4,4,3,3,3)$}
		\label{fig:rateprof43}
	\end{subfigure}
	
	\caption{Rate-memory trade-off of a shared cache network with $\Lambda=6$, $K=21$, $N=42$. }
	\label{fig:example2}
\end{figure*}

Corresponding to each distinct ordered pair in $\mathbf{G}$, the server generates $20$ independent random keys, each uniformly distributed over $[2^{F_s}]$ and is indexed by an ordered pair. The keys stored in each user's cache are:
\begin{equation*}
\begin{aligned}
\mathcal{Z}_1 & = \{\mathcal{K}_{(1,1)}, \mathcal{K}_{(2,1)}\}, \textrm{\hspace{0.05cm}} \mathcal{Z}_2  = \{\mathcal{K}_{(1,2)}, \mathcal{K}_{(2,2)}\}\\
\mathcal{Z}_3 & = \{\mathcal{K}_{(1,3)}, \mathcal{K}_{(2,3)}\}, \textrm{\hspace{0.05cm}} \mathcal{Z}_4  = \{\mathcal{K}_{(1,4)}, \mathcal{K}_{(2,4)}\}\\
\mathcal{Z}_5 & = \{\mathcal{K}_{(1,5)}, \mathcal{K}_{(2,5)}\}, \textrm{\hspace{0.05cm}} \mathcal{Z}_6  = \{\mathcal{K}_{(1,6)}, \mathcal{K}_{(2,6)}\}\\
\mathcal{Z}_7 & = \{\mathcal{K}_{(1,1)}, \mathcal{K}_{(3,1)}\}, \textrm{\hspace{0.05cm}} \mathcal{Z}_8  = \{\mathcal{K}_{(1,2)}, \mathcal{K}_{(3,2)}\}\\
\mathcal{Z}_9 & = \{\mathcal{K}_{(1,3)}, \mathcal{K}_{(3,3)}\}, \textrm{\hspace{0.05cm}} \mathcal{Z}_{10}   = \{\mathcal{K}_{(1,4)}, \mathcal{K}_{(3,4)}\}\\
\mathcal{Z}_{11} & = \{\mathcal{K}_{(1,5)}, \mathcal{K}_{(3,5)}\}, \textrm{\hspace{0.05cm}} \mathcal{Z}_{12}   = \{\mathcal{K}_{(2,1)}, \mathcal{K}_{(3,1)}\}\\
\mathcal{Z}_{13} & = \{\mathcal{K}_{(2,2)}, \mathcal{K}_{(3,2)}\}, \textrm{\hspace{0.05cm}} \mathcal{Z}_{14}   = \{\mathcal{K}_{(2,3)}, \mathcal{K}_{(3,3)}\}\\
\mathcal{Z}_{15} &  = \{\mathcal{K}_{(2,4)}, \mathcal{K}_{(3,4)}\},\textrm{\hspace{0.05cm}} \mathcal{Z}_{16}   = \{\mathcal{K}_{(1,1)}, \mathcal{K}_{(4,1)}\}\\
\mathcal{Z}_{17} & = \{\mathcal{K}_{(1,2)}, \mathcal{K}_{(4,2)}\}, \textrm{\hspace{0.05cm}} \mathcal{Z}_{18}   = \{\mathcal{K}_{(1,3)}, \mathcal{K}_{(4,3)}\}\\
\mathcal{Z}_{19} & = \{\mathcal{K}_{(2,1)}, \mathcal{K}_{(4,1)}\}, \textrm{\hspace{0.05cm}} \mathcal{Z}_{20}   = \{\mathcal{K}_{(2,2)}, \mathcal{K}_{(4,2)}\}\\
\mathcal{Z}_{21} & = \{\mathcal{K}_{(3,1)}, \mathcal{K}_{(4,1)}\}.
\end{aligned}
\end{equation*}

In the delivery phase, each user $k \in [21]$ requests a file $W_k$ from the server. Then, the messages transmitted are as follows:
\begin{equation*}
\begin{aligned}
  X_{(1,1)} & = S_{1,3} \oplus S_{7,2} \oplus S_{16,1} \oplus \mathcal{K}_{(1,1)} \\
  X_{(1,2)} & = S_{2,3} \oplus S_{8,2} \oplus S_{17,1} \oplus \mathcal{K}_{(1,2)} \\
  X_{(1,3)} & = S_{3,3} \oplus S_{9,2} \oplus S_{18,1} \oplus \mathcal{K}_{(1,3)} \\
  X_{(1,4)} & = S_{4,3} \oplus S_{10,2} \oplus \mathcal{K}_{(1,4)} \\
  X_{(1,5)} & = S_{5,3} \oplus S_{11,2} \oplus \mathcal{K}_{(1,5)} \\
  X_{(1,6)} & = S_{6,3} \oplus \mathcal{K}_{(1,6)} \\
  X_{(2,1)} & = S_{1,4} \oplus S_{12,2} \oplus S_{19,1} \oplus \mathcal{K}_{(2,1)} \\
  X_{(2,2)} & = S_{2,4} \oplus S_{13,2} \oplus S_{20,1} \oplus \mathcal{K}_{(2,2)} \\
  X_{(2,3)} & = S_{3,4} \oplus S_{14,2} \oplus \mathcal{K}_{(2,3)} \\
  X_{(2,4)} & = S_{4,4} \oplus S_{15,2} \oplus \mathcal{K}_{(2,4)}\\
   X_{(2,5)} & = S_{5,4} \oplus \mathcal{K}_{(2,5)}, \textrm{\hspace{0.02cm}} X_{(2,6)} = S_{(6,4)} \oplus \mathcal{K}_{(2,6)} \\  X_{(3,1)} & = S_{7,4} \oplus S_{12,3} \oplus S_{21,1}\oplus \mathcal{K}_{(3,1)}\\
    X_{(3,2)} & = S_{8,4} \oplus S_{13,3} \oplus \mathcal{K}_{(3,2)}\\
    X_{(3,3)} & = S_{9,4} \oplus S_{14,3} \oplus \mathcal{K}_{(3,3)}\\    
    X_{(3,4)} & = S_{10,4} \oplus S_{15,3} \oplus \mathcal{K}_{(3,4)}, \textrm{\hspace{0.02cm}} X_{(3,5)} = S_{11,4} \oplus \mathcal{K}_{(3,5)}\\ 
    X_{(4,1)} & = S_{16,4} \oplus S_{19,3} \oplus S_{21,2} \oplus \mathcal{K}_{(4,1)}\\ 
    X_{(4,2)} & = S_{17,4} \oplus S_{20,3} \oplus \mathcal{K}_{(4,2)} \\
    X_{(4,3)} & = S_{18,4}  \oplus \mathcal{K}_{(4,3)}. 
 \end{aligned}
\end{equation*}

Observe that each transmission is encrypted by a key that is available with all the users involved in that transmission, and each user is able to recover the remaining two shares of its requested file using its available cache contents. Since the Cauchy matrix is invertible, each user $k$ can recover its requested file $W_k$ from the obtained four shares.

The secrecy constraint in \eqref{secrecy} is also satisfied as each user has obtained only the new shares of its requested file, thus no information is gained about the remaining files. Also, note that the delivery scheme ensures secrecy of the files from an external adversary which observes the transmissions over the broadcast link. Thus, the decodability and the secrecy constraints in \eqref{decod} and \eqref{secrecy} are satisfied.

 \begin{figure}[t!]
 	\begin{center}
 		\captionsetup{justification=centering}
 		\includegraphics[width=\columnwidth]{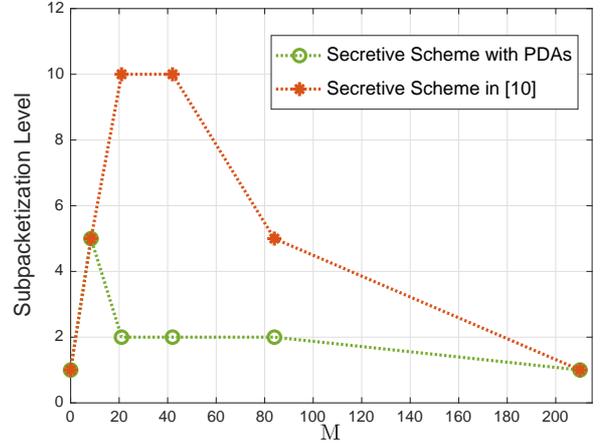}
 		\caption{Subpacketization level for a shared cache network with $\Lambda=6$, $K=21$, $N=42$. }
 		\label{fig:subpacket}
 	\end{center}
 \end{figure} 
 
 Since there are $20$ transmissions, each of normalized size $\frac{1}{(F-Z)} = \frac{1}{2}$, the rate obtained is $R_s = 20/2 =10$. For the same setting, if we employ the scheme in \cite{MeR}, the subpacketization level required is $10$ and the normalized rate achieved is $R_s = 8.4$. Thus, by utilizing the available PDA constructions, we could get more practically realizable secretive coded caching schemes without paying much in the required rate. 
 
 The rate-memory trade-off of a shared cache network with $\Lambda=6$, $K=21$ and $N=42$ is depicted in Fig.~\ref{fig:example2}. In Fig.~\ref{fig:rateprof6543}, the profile $\mathcal{L}=(6,5,4,3,2,1)$ is considered and the rates achieved by our scheme and the scheme in \cite{MeR} are plotted for comparison. The scheme in \cite{MeR} is defined for $\frac{M}{N} \in \{0, \frac{1}{5}, \frac{2}{4}, \frac{3}{3}, \frac{4}{2},{5}\}$. The remaining points are achievable by memory-sharing. At points $\frac{M}{N} \in \{\frac{2}{4}, \frac{3}{3}, \frac{4}{2}\}$, we make use of the PDAs from \cite{YCT} and \cite{CWZW} to obtain the corresponding secretive coded caching scheme. At the above points, the scheme obtained by our procedure requires less subpacketization level than the scheme in \cite{MeR} and is illustrated in Fig.~\ref{fig:subpacket} . At $\frac{M}{N} \in \{\frac{1}{5}, 5\}$, the PDAs described in Proposition~\ref{proposition} are used to derive the secretive coded caching schemes. Fig.~\ref{fig:rateprof43} shows the performance comparison for another profile $\mathcal{L}=(4,4,4,3,3,3)$.

 Next, we show that the secretive coded caching scheme for shared caches given in \cite{MeR} can be recovered using our procedure.

\begin{prop}
 For a $(\Lambda, K, M, N)$ shared caching problem with an association profile $\mathcal{L}$ and  $\frac{M}{M+N} \in \{ 0, \frac{1}{\Lambda}, \frac{2}{\Lambda}, \ldots,\frac{\Lambda-1}{\Lambda}\}$, let $t_s \triangleq \frac{\Lambda M}{M+N}$, the secretive coded caching scheme proposed in \cite{MeR} can be recovered using a $(\Lambda, F, Z, S)$ PDA where $F = \binom{\Lambda}{t_s}$, $Z = \binom{\Lambda-1}{t_s-1}$ and $S=\binom{\Lambda}{t_s+1}$.
 \label{proposition}
\end{prop}
\begin{proof}
	
	The scheme in \cite{MeR} is defined for $\frac{M}{M+N} \in \{ 0, \frac{1}{\Lambda}, \frac{2}{\Lambda}, \ldots,\frac{\Lambda-1}{\Lambda}\}$. Therefore, $t_s \in [0,\Lambda-1]$. Choose a $t_s \in [0, \Lambda-1]$ such that $F = \binom{\Lambda}{t_s} $. Each row of the PDA $\mathbf{P}$ is indexed by sets $\mathcal{T} \subset [\Lambda]$ , where $|\mathcal{T}|=t_s$. Each helper cache, $\lambda$ stores the shares $S_{n,\mathcal{T}}$ $\forall n \in [N]$, if $\lambda \in \mathcal{T}$. Thus, we obtain $Z = \binom{\Lambda-1}{t_s-1}$. The integer entries in the PDA are obtained by defining a bijective function $f$ from the $t_s+1$ sized subsets of $[\Lambda]$ to the set $[\binom{\Lambda}{t_s+1}]$ such that
	\begin{equation}
	p_{\mathcal{T},\lambda} = 
	\begin{cases}
	f(\mathcal{T}\cup \{\lambda\}), & \textrm{if \hspace{0.05cm}} \lambda \notin \mathcal{T}.\\
	\star, & \textrm{elsewhere}.
	\end{cases}
	\label{eq:pda_MN}
	\end{equation}
	
	After constructing the PDA $\mathbf{P}$, follow the procedures in Algorithms $2$ and $3$. Then, we will obtain the secretive coded caching scheme in \cite{MeR}. The reason behind the above PDA construction is that the secretive coded caching scheme in \cite{MeR} is derived from a non-secretive coded caching scheme for shared caches proposed in \cite{PUE}, which in fact, is derived from the MN scheme. Corresponding to every MN scheme, we can obtain a PDA. Thus, starting with a PDA which corresponds to the equivalent MN scheme would result in the secretive coded caching scheme for shared cache.
	\end{proof}

  \section{Lower Bound and Order-Optimality}
 \label{privateappendix}
 In this section, we derive a cut-set based lower bound on $R_s^*(M,M_U)$ and then show that the performance of the new secretive coded caching scheme for shared caches obtained using PDAs, is always within a multiplicative gap of $\Lambda$ from the above information theoretic lower bound. 
 
 Let $s$ be an integer such that $s\in \{1,2,\hdots,\min(K,\frac{N}{2})\}$. Consider the caches $\mathcal{Z}_{[1:s]}$. Corresponding to the demand vector $\mathbf{d}_1 = (1,2,\hdots,s,\phi,\hdots,\phi)$, where the first $s$ requests are for distinct files and the last $K-s$ requests can be for arbitrary files, the server makes transmission $X_1$. Using $X_{1}$ and the cache contents  $\mathcal{Z}_{[1:s]}$ and $Z_{[1:\lambda_s]}$, the files $W_{[1:s]}$ can be decoded. Similarly, using the same cache contents and the transmission corresponding to the demand vector $\mathbf{d}_2 = (s+1,s+2,\hdots,2s,\phi,\hdots,\phi)$, the files $W_{[s+1:2s]}$ can be decoded. Therefore, considering $\floor{N/s}$ such demand vectors, there will be $\floor{N/s}$ transmissions, $X_{[1:\floor{N/s}]}$. Using $\mathcal{Z}_{[1:\lambda_s]}$, $Z_{[1:s]}$, and from the transmissions $X_{[1:\floor{N/s}]}$, the files $W_{[1:s\floor{N/s}]}$ can be decoded. Let $X_\ell$ denote the server transmission corresponding to the $\ell^{th}$ demand instance, where $\ell\in [1:\floor{N/s}]$. For $k\in [s]$, we define $d_k^\ell\triangleq (\ell-1)s+k$ and $\widehat{W} \triangleq W_{[1:s\floor{N/s}]\backslash d_k^\ell}$. Then we have,
 
 \begin{subequations}
 	\begin{align}
 	s\floor{N/s}-1 &= H(\widehat{W})\notag\\
 	&= I(\widehat{W};X_{[1:\floor{N/s}]},\mathcal{Z}_{[1:s]},Z_{[1:\lambda_s]})\notag\\
 	&\qquad+H(\widehat{W}|X_{[1:\floor{N/s}]},\mathcal{Z}_{[1:s]},Z_{[1:\lambda_s]})\notag\\
 	&= I(\widehat{W};X_{[1:\floor{N/s}]},\mathcal{Z}_{[1:s]},Z_{[1:\lambda_s]})\label{privacyc}\\
 	&= I(\widehat{W};X_\ell,\mathcal{Z}_k,Z_{\lambda_k}) \textrm{\hspace{0.1cm}}+\notag\\
 	&\quad I(\widehat{W};X_{[1:\floor{N/s}]\backslash\{\ell\}},\mathcal{Z}_{[1:s]\backslash\{k\}}, \notag \\  
 	& \qquad \qquad \qquad \qquad Z_{[1:\lambda_s]\backslash\{\lambda_k\}}|X_\ell,\mathcal{Z}_k,Z_{\lambda_k})\notag\\
    &= I(\widehat{W};X_{[1:\floor{N/s}]\backslash\{\ell\}},\mathcal{Z}_{[1:s]\backslash\{k\}}, \notag \\ 
    & \qquad \qquad Z_{[1:\lambda_s]\backslash\{\lambda_k\}}|X_\ell,\mathcal{Z}_k,Z_{\lambda_k})\label{privacye}\\
 	&\leq H(X_{[1:\floor{N/s}]\backslash\{\ell\}},\mathcal{Z}_{[1:s]\backslash\{k\}},Z_{[1:\lambda_s]\backslash\{\lambda_k\}})\notag\\
 	&\leq \sum_{\substack{j=1 \\ j\neq \ell}}^{\floor{N/s}} H(X_j)+\sum_{\substack{j=1 \\ j\neq k}}^s H(\mathcal{Z}_j)++\sum_{\substack{j=1 \\ j\neq \lambda_k}}^{\lambda_s} H(Z_j) \notag\\
 	&\leq (\floor{N/s}-1)R_s^*(M)+(s-1)M_U\notag\\
 	&\qquad+(\lambda_s-1)M,\notag
 	\end{align}
 \end{subequations}
 where \eqref{privacyc} follows from the decodability condition \eqref{decod}, and \eqref{privacye} follows from the secrecy condition \eqref{secrecy}.
 
 Upon rearranging the terms, we obtain,
 \begin{equation*}
 R_s^*(M,M_U) \geq \frac{s\floor{N/s}-1-(\lambda_s-1)M-(s-1)M_U}{\floor{N/s}-1}.
 \end{equation*}
 
 Optimizing over all possible choices of $s$, we have the lower bound on $R_s^*(M,M_U)$,
 \begin{equation*}
 \begin{aligned}
 & R_s^*(M,M_U) \geq \\
 & \max_{s\in \{1,2,\hdots,\min(\frac{N}{2},K)\}} \small{ \frac{s\floor{N/s}-1-(\lambda_s-1)M-(s-1)M_U}{\floor{N/s}-1}}.
 \end{aligned}
 \end{equation*}
 
 This completes the proof of Theorem \ref{thmsecrecy}.\hfill $\blacksquare$ 

	\begin{figure}[t!]
 	\begin{center}
 		\captionsetup{justification=centering}
 		\includegraphics[width=\columnwidth]{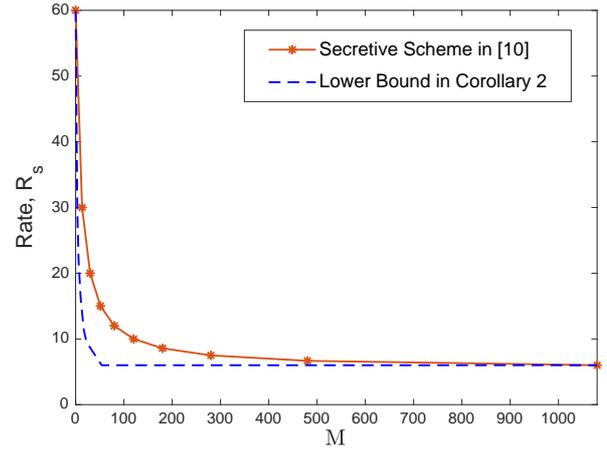}
 		\caption{Rate-memory trade-off of shared cache network with $\Lambda=10$, $K=60$, $N=120$ with uniform association profile.}
 		\label{fig:ratemem}
 	\end{center}
 \end{figure}

 In our setting, $M_U$ is fixed as unity. Therefore, the above lower bound expression reduces to 
\begin{equation*}
R^{*}_s (M) \geq \underset{s \in \{1,2,\ldots, \textrm{min}(K,\frac{N}{2})\}}{\textrm{max}} s-\frac{(\lambda_s-1)M}{\floor{\frac{N}{s}}-1}.
\end{equation*}
Assume $N \geq 2K$, let $s=\mathcal{L}_1$, then we obtain
\begin{equation*}
 R^{*}_s(M) \geq \mathcal{L}_1.
 \end{equation*}
 
 In a $(\Lambda, F, Z, S)$ PDA, $S \leq \Lambda(F-Z)$. If $S = \Lambda(F-Z)$, then $R_s = \frac{K(F-Z)}{(F-Z)}=K$, which implies $R_s(M) \leq K$ in any secretive coded caching scheme obtained using PDAs. Consequently, 
 
 \begin{equation}
 \frac{R_s(M)}{R^{*}_s(M)}  \leq \frac{K}{\mathcal{L}_1}. 
 \label{eq:bound}  							
 \end{equation}

Since, $\mathcal{L}_1 \geq \frac{K}{\Lambda}$, \eqref{eq:bound} reduces to
\begin{equation*}
\frac{R_s(M)}{R^{*}_s(M)}  \leq \Lambda.  							
\end{equation*}

This completes the proof of Theorem~~\ref{thm3}. \hfill $\blacksquare$

In Fig.~\ref{fig:ratemem}, the performance of the  scheme in \cite{MeR} is compared against the cut-set lower bound obtained in Corollary~\ref{corr}. The optimality gap presented in Theorem~\ref{thm3} also holds for the secretive coded caching scheme given in \cite{MeR}.

\section{Conclusion}
\label{sec:con}
In this work, we developed a procedure to derive secretive coded caching schemes for shared caches using PDAs. The advantage of our scheme is that the extensive literature on PDA constructions can be leveraged to obtain more practically realizable secretive coded caching schemes for shared caches. The subpacketization level reduction achieved by our schemes is obtained by paying in the rate. Therefore, the PDAs must be chosen depending on the file size and the allowable subpacketization level. We have also derived a lower bound based on cut-set arguments and characterized the performance of the achievable scheme.

 \section*{Acknowledgement}
 This work was supported partly by the Science and Engineering Research Board (SERB) of Department of Science and Technology (DST), Government of India, through J.C Bose National Fellowship to B. Sundar Rajan.

\end{document}